\documentclass[journal]{IEEEtran}
\usepackage{amsmath,amssymb,amsfonts}
\usepackage{graphicx}
\usepackage{cite}
\usepackage{hyperref}
\usepackage{bm} 
\usepackage{booktabs} 
\usepackage{siunitx}  
\usepackage{booktabs} 
\usepackage{siunitx}  
\usepackage{graphicx} 
\usepackage{subcaption} 
\usepackage{algorithmic}
\usepackage{algorithm}
\usepackage{tabularx}
\usepackage{optidef}

\usepackage{amsthm}
\usepackage{svg}
\theoremstyle{plain}
\newtheorem{definition}{Definition}
\newtheorem{proposition}{Proposition}
\newtheorem{remark}{Remark}
\usepackage[font=small]{caption}

\ifCLASSINFOpdf
\else
\fi
\makeatletter
\def\sublargesize{\@setfontsize{\sublargesize}{10pt}{12pt}}
\makeatother

\begin{document}
%
\title{{\fontsize{22}{24}\selectfont Grid-Interactive Thermal Management of AI Data Centers via Contextual Distributionally Robust Optimization}}

\author{
Jiachen Shen,~\IEEEmembership{Student Member,~IEEE,}
	Jian Shi,~\IEEEmembership{Senior Member,~IEEE,}
		  Yijie Yang,~\IEEEmembership{Member,~IEEE,}
		Chenye Wu,~\IEEEmembership{Senior Member,~IEEE,}
        Dan Wang,~\IEEEmembership{Senior Member,~IEEE,}
        Ju Bin Song,~\IEEEmembership{Member,~IEEE,}
		and Zhu Han,~\IEEEmembership{Fellow,~IEEE}
		\vspace{-10mm}






}

\maketitle

\begin{abstract}
Thermal management in AI data centers is increasingly challenged by bursty workloads and uncertain heat generation. To prevent thermal violations, existing cooling strategies either enforce conservative, rigid bounds that severely limit grid responsiveness, or rely on forecast-driven controllers that perform poorly under AI workload uncertainty and distribution shifts. To overcome the above challenges, this paper proposes a Contextual Distributionally Robust Optimization (CDRO) framework for grid-interactive cooling control. Unlike standard DRO with fixed ambiguity sets, the proposed approach dynamically adapts the Wasserstein radius using real-time AI and grid context. This safely shrinks uncertainty bounds during stable regimes, unlocking deep demand-side flexibility. Theoretically, we formulate the control as an infinite-dimensional inf–sup problem, derive an exact tractable reformulation for the Wasserstein worst-case expected-cost term, and then derive a tractable conservative deterministic counterpart for the Distributionally Robust Conditional Value at Risk (DR-CVaR) thermal safety constraint. Solved via a scalable nested Alternating Direction Method of Multipliers (ADMM) algorithm, the CDRO controller achieves near-zero thermal violations under extreme workload spikes in high-fidelity EnergyPlus co-simulations. Simultaneously, it reduces the operational cost premium of robustness by approximately 13.7 percentage points relative to standard Min–Max Model Predictive Control (MPC).

\end{abstract}

\IEEEpeerreviewmaketitle
\vspace{-2mm}
\section{Introduction}
\label{sec:introduction}

\subsection{Motivation and Challenges}
The rapid growth of Generative Artificial Intelligence (GenAI) has significantly changed the operational landscape of data centers \cite{AIIndex2024}. This elevates thermal management to a critical infrastructure bottleneck alongside computational capacity. Modern AI accelerators, like NVIDIA's H100 GPUs, exhibit a Thermal Design Power (TDP) exceeding 700W per chip. This pushes rack power densities beyond 50kW, which is ten times that of traditional CPU-based servers \cite{iea2024electricity}. This rapid rise in heat flux requires powerful cooling capacities. As a result, the cooling system becomes the largest auxiliary energy consumer (accounting for 30\%--40\% of total facility energy) and the most critical subsystem for hardware reliability. Also, the central operational challenge has shifted from general energy scheduling to the precise control of cooling loads: maintaining thermal safety under extreme heat densities while attempting to adjust these massive cooling loads to respond to volatile power grid signals (e.g., locational marginal prices (LMP) or carbon intensity).

Moreover, the inflexible operational strategies of legacy cooling systems are increasingly in conflict with the decarbonization goals of the energy sector. As the grid integrates higher shares of variable renewable energy, electricity prices and carbon intensity become highly volatile \cite{Chen2024_FlowBasedCarbonAccounting, Chen2024_CarbonAwareDR}. Static cooling strategies (e.g., constant setpoint or Proportional-Integral-Derivative (PID) control), which consume power regardless of these external signals, not only incur excessive operational costs but also worsen grid congestion during peak hours \cite{Takci2025_DCFlexibility}. There is an urgent need to utilize the unused flexibility within the cooling infrastructure. This allows it to act as a responsive demand-side resource that can absorb renewable generation when it is abundant and shed load when the grid is stressed.

To address this, the concept of the grid-interactive data center has emerged recently. Here, cooling infrastructure functions as a flexible thermal battery, capable of pre-cooling during periods of renewable abundance and shedding load during grid stress \cite{Takci2025_DCFlexibility}. While promising, transforming AI data center cooling into a flexible asset faces two distinct control challenges. These stem directly from the unique physical characteristics of AI cooling loads:

\textit{1) Timescale Mismatch between Thermal Shock and Cooling Dynamics:} 
The first and most formidable challenge is the significant gap between the volatility of AI thermal loads and the response speed of industrial cooling infrastructure. AI inference and training workloads are characterized by extreme burstiness. A sudden surge in matrix multiplication operations can ramp up heat generation from idle to peak in milliseconds, thus creating immediate localized thermal shocks. In contrast, the heat rejection chain is governed by slow thermo-fluid dynamics. This infrastructure relies on heavy mechanical components like chillers and pumps, which operate on time constants ranging from minutes to hours \cite{Cao2022_IDC_Flexibility}. This operational mismatch renders the cooling system naturally slow. Traditional controllers cannot anticipate these millisecond-level heat spikes. They often fail to ramp up cooling capacity in time. This leads to rapid heat accumulation and potential thermal runaway.

\textit{2) Over-Cooling Trap vs. Demand Flexibility:} 
Due to the aforementioned thermal risks, current industry practices rely on extreme conservatism, known as over-cooling. This involves maintaining setpoints far lower than necessary to create a safety buffer against unforeseen load spikes. While physically safe, this static conservatism rigidly locks the cooling power at high levels. It effectively eliminates the facility's ability to participate in Demand Response \cite{Takci2025_DCFlexibility}. This creates a fundamental conflict in cooling load management. To provide flexibility, the system must reduce cooling power when grid prices are high. This requires operating closer to its thermal limits. However, doing so with standard controllers (like PID or deterministic MPC) creates risks. It exposes mission-critical hardware to unacceptable overheating caused by the unpredictable nature of AI tasks. Thus, the core problem is not merely optimizing energy, but designing a control framework that can \textit{dynamically} balance thermal safety margins and grid flexibility based on real-time contexts.
\vspace{-4mm}
\subsection{Literature Review}
The need for efficient cooling load management has driven significant research into advanced data center control strategies. This drive is closely aligned with the Green AI paradigm \cite{schwartz2020green}, yet specifically focuses on the thermo-mechanical layer of infrastructure operation.

\textit{1) Limitations of Existing Cooling Control:} The industry standard, PID control, is purely reactive and ill-suited for the high-density cooling requirements of AI clusters. It lacks the foresight to manage heat accumulation, often resulting in oscillatory behavior and significant energy waste \cite{ControlSurvey_DataCenter_2020}. Model Predictive Control advances this by utilizing physics-based models to optimize cooling actions over a receding horizon \cite{Nassif2012_MPC_DC}. However, deterministic MPC is highly sensitive when applied to AI cooling loads. It relies on point forecasts of IT power, which are practically impossible to predict accurately due to the stochastic arrival of inference queries \cite{Gao2022_WorkloadPrediction}. A minor prediction error during a cooling load shedding event can lead to immediate Service Level Agreement (SLA) violations \cite{ Gao2022_WorkloadPrediction}. While Reinforcement Learning (RL) offers a model-free alternative \cite{ Li2019_RLCooling}, standard RL agents lack the rigorous safety guarantees required to operate critical cooling plants, often requiring unsafe exploration phases \cite{Wang2024_SafeDCCooling}.

\textit{2) Optimization under Uncertainty and Distributional Robustness:} Uncertainty in power and energy systems is commonly handled via stochastic programming (SP), which typically assumes known probability distributions and can become sensitive to distribution shift, and via robust optimization (RO), which enforces feasibility under bounded uncertainty sets at the cost of conservatism. In the receding-horizon control literature, several uncertainty-aware MPC families have been developed to balance performance and safety: tube/min--max MPC enforces robust feasibility against bounded disturbances \cite{Kiaei2018TubeBasedMPCTSG,Xie2021EnhancedTubeMPCTSG}, chance-constrained MPC enforces probabilistic safety via quantile/violation-rate constraints under assumed or empirically estimated error models \cite{Ravichandran2018ChanceConstraintsTSG}, and risk-sensitive MPC (e.g., Conditional Value at Risk (CVaR) based designs) explicitly trades expected cost for tail-risk mitigation \cite{Rosewater2020RiskAverseMPCTSG}. More recently, distributionally robust formulations have gained traction by optimizing against an ambiguity set of distributions (e.g., Wasserstein balls), yielding DR-MPC variants that protect against distributional shift without committing to a single parametric forecast-error distribution \cite{Nguyen2023DRMPC_EVCS_TSG,Li2024WDRMPCTSG}; related distributed robust MPC formulations have also been studied for multi-microgrid coordination under uncertainty \cite{Zhao2022DistributedRobustMPCTSG}. 

However, these approaches are typically non-contextual: the level of conservatism is governed by static bounds, fixed quantiles/CVaR levels, or a globally chosen ambiguity radius, which can be inefficient when uncertainty regimes change rapidly with internal compute states and external grid conditions. This motivates a contextual distributionally robust framework that leverages multi-modal context (e.g., queue/workload states and grid volatility signals) to adapt robustness online, enabling calibrated conservatism while maintaining thermal safety under deep and regime-shifting uncertainty.
\vspace{-6mm}
\subsection{Our Contributions}
To address the research gaps above, we propose an online CDRO framework for carbon-aware thermal management in AI data centers. We formulate the cooling control as an inf–sup optimization problem over an infinite-dimensional space of probability measures, constrained by non-convex thermo-fluid dynamics. This problem is computationally intractable to solve directly because it tightly couples deep distributional uncertainty with non-convex Mixed-Integer Second-Order Cone Program (MISOCP) plant constraints, including cubic fan power laws and bi-quadratic chiller efficiencies. To solve it in real-time dispatch, we derive exact tractable reformulations and an ADMM-based decomposition algorithm. Specifically, this paper makes three main contributions:
\begin{itemize}
    \item We propose the first-of-its-kind online CDRO framework for grid-interactive AI data center cooling to our knowledge. By dynamically adapting the Wasserstein radius using multi-modal real-time context (e.g., workload/queue states and grid volatility signals), the controller calibrates conservatism across regimes while enforcing a Distributionally robust-CVaR thermal management constraint to control tail overheating risk and unlock grid-responsive flexibility.


    \item We derive tractable deterministic reformulations for the key distributionally robust terms. In particular, we obtain an exact reformulation of the Wasserstein worst-case expected-cost term via strong duality, and a tractable conservative deterministic counterpart for the DR-CVaR thermal safety constraint to avoid ad hoc worst-case constraint tightening.
    
    
    \item  We design a nested solution algorithm leveraging the ADMM to decouple subsystem optimizations and efficiently resolve the coupled MISOCP constraints. Furthermore, we demonstrate that the computational complexity of the proposed online controller scales linearly with the prediction horizon and spatial dimensions, guaranteeing theoretical tractability and highly scalable real-time execution for standard 5-minute control intervals.

\end{itemize}

The remainder of this paper is organized as follows. Section \ref{sec:problem_formulation} details the thermo-fluid dynamics and problem formulation. Section \ref{sec:methodology} describes the proposed CDRO methodology. Section \ref{sec:experimental_evaluation} presents the experimental setup and numerical results, followed by the conclusion in Section \ref{sec:conclusion}.

\begin{figure}[t]
    \centering \includegraphics[width=0.45\textwidth]{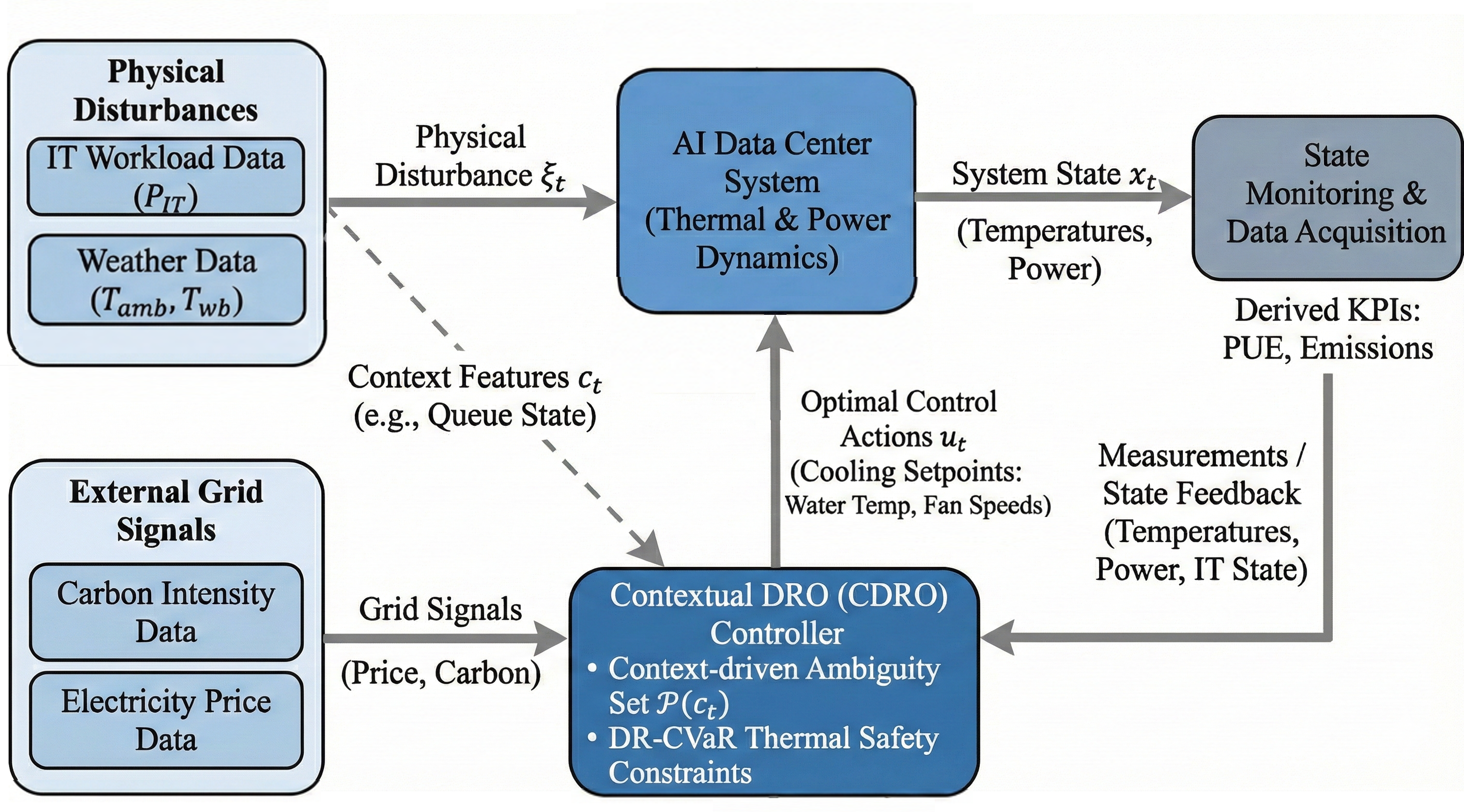}
    \caption{System Framework: The integrated control loop from external grid and ambient temperature inputs to internal data hall and cooling plant dynamics.}
    \label{fig:1}
    \vspace{-5mm}
\end{figure}

\vspace{-5mm}

\section{Problem Formulation}
\label{sec:problem_formulation}
\vspace{-2mm}
In this section, we build a mathematical framework for a grid-interactive AI data center. We first model the system using control-oriented thermal entities and cooling plant dynamics. Then, we formulate the supervisory scheduling task as a CDRO problem. This formulation combines deterministic physical constraints, a distributionally robust objective, and a DR-CVaR hotspot safety constraint. Table \ref{tab:notation_all_single} summarizes the key mathematical notation.

\vspace{-5mm}
\subsection{System Dynamics and Physical Modeling}
\label{ssec:physical_model}

To ensure real-time computational efficiency, we employ discrete-time lumped-parameter models with a control interval of $\Delta t$ (e.g., 5 minutes) that capture essential non-linear dynamics without detailed computational fluid dynamics simulations. For spatial granularity, the control-oriented thermal dynamics are defined over a set of instrumented thermal entities (e.g., racks or cooling zones) indexed by $z \in \mathcal{L}_z$. In parallel, server-level telemetry channels indexed by $s \in \mathcal{L}_s$ are used for hotspot risk evaluation, with a fixed mapping $g:\mathcal{L}_s \rightarrow \mathcal{L}_z$ assigning each server to its hosting entity.

\subsubsection{Control-Oriented Thermal Entity Dynamics}
The core temperature of a controlled thermal entity (rack/zone) ($T_{\text{core},z}$) is a critical state. We model its dynamics using a first-order Resistor-Capacitor (RC) thermal circuit model \cite{Wang2008_RC_Model_DC}. The continuous-time dynamics are discretized using a forward Euler method, resulting in the following linear state-space representation for each controlled entity $z \in \mathcal{L}_z$:

\vspace{-3mm}
{\small
\begin{equation}
    T_{\text{core},z}[t+1] = A_z T_{\text{core},z}[t] + B_z P_{\text{IT},z}[t] + C_z T_{\text{in},z}[t],
    \label{eq:server_rc}
\end{equation}
}

\noindent where $A_z, B_z, C_z$ are coefficients derived from the thermal capacitance, heat resistance, and $\Delta t$. $P_{\text{IT},z}[t]$ is the aggregated and highly variable power consumed by the IT equipment. $T_{\text{in},z}[t]$ is the inlet temperature at the controlled entity. We specifically define $T_{\text{core},z}[t]$ as the representative or control-oriented hotspot temperature state for thermal entity $z$. The server-level hotspot evaluation is handled separately via telemetry (Section \ref{ssec:cdro_model}).

\subsubsection{Room Air and Cooling Loop Dynamics}
We model the cold aisle as a single, well-mixed air volume. Its average temperature, which serves as the shared inlet temperature $T_{\text{in}}[t]$, is determined by an energy balance between the heat removed by the Computer Room Air Conditioning (CRAC) units and the heat recirculated from the hot aisle, following established lumped-capacitance principles \cite{Tang2006_HeatRecirculation}. The dynamics are given by:

\vspace{-3mm}
{\small
\begin{equation}
    T_{\text{in}}[t+1] = (1-\beta) T_{\text{in}}[t] + \beta T_{\text{sup,mix}}[t] + \eta \sum_{z \in \mathcal{L}_z} P_{\text{IT},z}[t].
\end{equation}
}
\vspace{-3mm}

\noindent Here, $P_{\text{IT},z}[t]$ denotes the aggregate IT power load of the servers assigned to thermal entity $z$, meaning $P_{\text{IT},z}[t] = \sum_{s: g(s)=z} P_{\text{IT},s}[t]$. $T_{\text{sup,mix}}[t]$ is the mixed supply air temperature from all CRACs, and $\beta, \eta$ are mixing and recirculation coefficients. The mixed supply air temperature, $T_{\text{sup,mix}}[t]$, represents the mass-flow-weighted average temperature of the air supplied by all active CRAC units, and is defined as:

{\small
\begin{equation}
    T_{\text{sup,mix}}[t] = \frac{\sum_{i \in \mathcal{I}} m_{\text{air},i}[t] \cdot T_{\text{sup},i}[t]}{\sum_{i \in \mathcal{I}} m_{\text{air},i}[t]}.
\end{equation}
}
The supply air temperature from a CRAC unit $i \in \mathcal{I}$, $T_{\text{sup},i}[t]$, depends on its heat exchanger effectiveness $\epsilon_i$ and the chilled water supply temperature $T_{\text{chw,sup}}[t]$. The heat removed by the CRACs, $Q_{\text{CRAC}}[t]$, must equal the total load on the chilled water plant.

\begin{table}[t]
\centering
\caption{\small Sets, Parameters, and Variables}
\label{tab:notation_all_single}
\begin{tabularx}{\columnwidth}{p{2.8cm} X}
\hline
\textbf{Symbol} & \textbf{Description} \\
\hline
\multicolumn{2}{l}{\textit{Sets and Indices}} \\
$\mathcal{T}, \mathcal{I}$ & Set of discrete time steps and CRAC units \\
$\mathcal{L}_z$ & Set of controlled thermal entities (racks/zones), $z \in \mathcal{L}_z$ \\
$\mathcal{L}_s$ & Set of server-level telemetry entities, $s \in \mathcal{L}_s$ \\
$Z_{\text{hot}}[t]$ & Hotspot-induced controlled entities mapped from server telemetry \\
$\mathcal{L}_{\text{hot}}[t]$ & Critical hotspot servers (top $M_{\text{hot}}$ telemetry channels) \\
$g(s)$ & Mapping from server $s$ to its hosting thermal entity $z$ \\
\hline
\multicolumn{2}{l}{\textit{Model Parameters}} \\
$T_{\text{core,max}}$ & Soft temperature threshold for penalty (°C) \\
$T_{\text{core,crit}}$ & Hard critical accelerator temperature limit (°C) \\
$w_{\text{th}}$ & Penalty cost per degree of thermal violation (\$/°C) \\
$\kappa_{\text{CO}_2}$ & Carbon price coefficient (\$/kgCO$_2$) \\
$\varepsilon$ & DR-CVaR tail probability level (unitless fraction; $\varepsilon = 0.05$ means 5\%) \\
\hline
\multicolumn{2}{l}{\textit{Variables and Functions}} \\
$\mathbf{u}[t]$ & Aggregated cooling control actions vector \\
$T_{\text{hot}}^{\text{proxy}}[t]$ & Zone-level proxy state representing the hotspot peak temperature \\
\hline
\multicolumn{2}{l}{\textit{Uncertain Parameters and Context}} \\
$\bm{\xi}[t]$ & Uncertain parameter realizations \\
$\hat{\bm{\xi}}[t:t+H]$ & Rolling forecasts generated using $\mathcal{I}_t$ (non-anticipative) \\
$\mathbf{c}[t]$ & Real-time context vector \\
\hline
\end{tabularx}
\vspace{-4mm}
\end{table}

\subsubsection{CRAC Unit Thermal-Fluid Dynamics}
The performance of each CRAC unit $i \in \mathcal{I}$ is described by the following relationships. The mass flow rate of air, $m_{\text{air},i}[t]$, is directly proportional to the fan speed decision $s_i[t]$:

{\small
\begin{equation}
    m_{\text{air},i}[t] = m_{\text{air},i}^{\text{rated}} \cdot s_i[t],
\end{equation}
}
where $m_{\text{air},i}^{\text{rated}}$ is the fan's rated mass flow rate at full speed.

The supply air temperature, $T_{\text{sup},i}[t]$, is the temperature of the cold air leaving the CRAC unit. It is determined by the entering return air temperature, $T_{\text{ret},i}[t]$, and the chilled water supply temperature, $T_{\text{chw,sup}}[t]$, via the standard $\epsilon$-NTU heat exchanger model \cite{KaysLondon_CompactHeatExchangers}:

{\small
\begin{equation}
    T_{\text{sup},i}[t] = T_{\text{ret},i}[t] - \epsilon_i \cdot (T_{\text{ret},i}[t] - T_{\text{chw,sup}}[t]),
\end{equation}
}
where $\epsilon_i$ is the heat exchanger effectiveness parameter (a constant between 0 and 1). The chilled water temperature $T_{\text{chw,sup}}[t]$ is assumed to be equal to our decision variable, the setpoint $T_{\text{chw,set}}[t]$.

The return air temperature, $T_{\text{ret},i}[t]$, is the temperature of the hot air from the hot aisle entering the CRAC unit. In our lumped-parameter model, we approximate this as the cold aisle temperature plus the temperature gain from the total IT load:

{\small
\begin{equation}
    T_{\text{ret,i}}[t] \approx T_{\text{in}}[t] + \frac{\eta_{\text{capture}} \sum_{z \in \mathcal{L}_z} P_{\text{IT},z}[t]}{\sum_{i \in \mathcal{I}} m_{\text{air},i}[t] \cdot c_{p,air}},
\end{equation}
}
where $\eta_{\text{capture}}$ is a heat capture effectiveness parameter that accounts for effects like air bypass and recirculation and $c_{p,air}$ is the specific heat capacity of air (a constant, approx. 1,005 J/kg·K).

\subsubsection{Cooling Plant Energy Model}
The cooling plant consists of CRAC fans, chillers, cooling towers, and pumps. Their power consumption is highly non-linear.
\begin{itemize}
    \item{\em CRAC Fans:} The power of the fan in CRAC unit $i$, $P_{\text{fan},i}[t]$, follows the fan laws, scaling cubically with its normalized speed $s_i[t]$:

    {\small
    \begin{equation}
        P_{\text{fan},i}[t] = P_{\text{fan},i}^{\text{rated}} \cdot (s_i[t])^3.
    \end{equation}
}
    \item{\em  Chillers:} The chiller power, $P_{\text{chiller}}[t]$, is the ratio of the thermal load it serves, $Q_{\text{CRAC}}[t]$, to its Coefficient of Performance (COP). The COP is a strong function of the chilled water supply setpoint $T_{\text{chw,set}}[t]$ and the temperature of the condenser water $T_{\text{cond,in}}[t]$ entering from the cooling tower. We use a standard bi-quadratic DOE-2 model \cite{Hydeman2002_ChillerModel}:
    
    {\small
    \begin{equation}
        \text{COP}[t] = \frac{Q_{\text{CRAC}}[t]}{P_{\text{chiller}}[t]} = f_{\text{COP}}(T_{\text{chw,set}}[t], T_{\text{cond,in}}[t]).
        \label{eq:cop}
    \end{equation}
    }
    \item{\em Cooling Towers:} The tower rejects the data center's waste heat ($Q_{\text{CRAC}}[t] + P_{\text{chiller}}[t]$) to the atmosphere. Its performance is determined by how closely its output water temperature, $T_{\text{cond,out}}[t]$, can approach the ambient wet-bulb temperature $T_{\text{wb}}[t]$. The tower fan power, $P_{\text{tower}}[t]$, also follows a cubic scaling law. The energy balance implies:

    \vspace{-3mm}
    {\small
    \begin{equation}
    \label{eq:tower}
        T_{\text{cond,in}}[t] = T_{\text{cond,out}}[t] \approx T_{\text{wb}}(\xi[t]) + f_{\text{approach}}(s_{\text{tower}}[t]).
    \end{equation}
    }
    
    \item{\em Pumps:} The power for chilled water and condenser water pumps, $P_{\text{pump}}[t]$, is modeled as a polynomial function of the required water flow rate, which in turn depends on the total thermal load.
\end{itemize}
\vspace{-4mm}
\subsection{Contextual Distributionally Robust Optimization Model}
\label{ssec:cdro_model}

Based on the physical dynamics, we formulate the supervisory cooling control problem. At each decision epoch $t$, the objective is to determine an optimal control sequence $\mathbf{u}$ over a finite horizon $\mathcal{T}$, minimizing the worst-case expected operational cost conditioned on the real-time context $\mathbf{c}[t]$ and uncertain parameters $\bm{\xi}[t]$. Specifically, the control actions vector $\mathbf{u}[t]$ comprises the normalized CRAC fan speeds $s_i[t]$, the cooling tower fan speed $s_{\text{tower}}[t]$, and the chilled water supply setpoint $T_{\text{chw,set}}[t]$. The uncertain parameter vector $\bm{\xi}[t]$ captures the highly variable aggregated IT load $P_{\text{IT},z}[t]$, ambient weather conditions $T_{\text{amb}}[t]$ and $T_{\text{wb}}[t]$, real-time electricity prices $C_{\text{elec}}[t]$, and grid carbon intensity $CI_{\text{grid}}[t]$.

\subsubsection{Contextual Ambiguity Set Construction}
To rigorously define the online decision boundaries and prevent information leakage, we first specify the available information set $\mathcal{I}_t$ at decision epoch $t$:

\begin{definition}[Online Information Set]
Let $\mathcal{I}_t$ denote the filtration containing all information strictly available up to time $t$:

\vspace{-3mm}
{\small
\begin{equation}
    \mathcal{I}_t := \{ \mathbf{T}_{\text{state}}[0:t], \mathbf{u}[0:t-1], \bm{\xi}_{\text{obs}}[0:t], \hat{\bm{\xi}}[t:t+H] \},
\end{equation}
}
where $\mathbf{T}_{\text{state}}[0:t]$ denotes the historical system states (e.g., measured core and return temperatures), $\bm{\xi}_{\text{obs}}[0:t]$ represents historical realizations of uncertain parameters, and $\hat{\bm{\xi}}[t:t+H]$ represents rolling forecasts generated by a predictive model trained strictly on data prior to $t$.

\end{definition}
Importantly, to manage the specific variability of AI workloads and grid signals, we explicitly define the context vector $\mathbf{c}[t]$ as a causal mapping of this information set, $\mathbf{c}[t] = \phi(\mathcal{I}_t)$. It combines internal computational features and external grid features:

\vspace{-3mm}
{\small
\begin{equation}
    \mathbf{c}[t] = [\mathbf{f}_{\text{queue}}[t], \mathbf{f}_{\text{grid}}[t]]^\top,
\end{equation}
}
where $\mathbf{f}_{\text{queue}}$ represents AI job scheduler metrics (e.g., job size, queue length), and $\mathbf{f}_{\text{grid}}$ includes forecasts of real-time carbon intensity and electricity prices. Specifically, $\mathbf{f}_{\text{grid}}$ incorporates publicly available Day-Ahead (DA) prices and short-term volatility nowcasts. For instance, the DA-RT spread feature utilizes the committed DA price and the \textit{current observed} Real-Time price, ensuring no future ground-truth values are leaked into the decision process. We construct a Wasserstein ambiguity set $\mathcal{P}(\mathbf{c}[t])$ based on the conditional distribution of prediction errors derived from $\mathcal{I}_t$.

\subsubsection{Objective Function}
The objective is to minimize the total operational cost in USD. To ensure dimensional consistency, we explicitly monetize carbon emissions and thermal risks using conversion factors. The objective is expressed as:

\vspace{-2mm}
{\small
\begin{equation}
    \mathcal{J}(\mathbf{u}) =  \sup_{P_\xi \in \mathcal{P}(\mathbf{c})} \mathbb{E}_{P_\xi} \left[\sum_{t \in \mathcal{T}} \text{Cost}[t]\right],
    \label{eq:obj}
\end{equation}
}
where the single-stage cost, $\text{Cost}[t]$, is defined in dollars (\$):

\vspace{-2mm}
{\small
\begin{equation}
\label{eq:objective_full}
\begin{aligned}
    \text{Cost}[t] = & \underbrace{C_{\text{elec}}[t] \cdot E_{\text{total}}[t]}_{\text{Energy Cost}} + \underbrace{\kappa_{\text{CO}_2} \cdot (CI_{\text{grid}}[t] \cdot E_{\text{total}}[t])}_{\text{Carbon Cost}} \\
    & + \underbrace{w_{\text{th}} \sum_{z \in \mathcal{L}_z} \max\{0, T_{\text{core},z}[t] - T_{\text{core,max}}\}}_{\text{Thermal Risk Penalty}}.
\end{aligned}
\end{equation}
}
Here, $E_{\text{total}}[t] = P_{\text{total}}[t] \cdot \Delta t$ represents the energy consumption over interval $\Delta t$, where $P_{\text{total}}[t]$ sums the power of chillers, towers, fans, and IT equipment. The parameter $\kappa_{\text{CO}_2}$ is the carbon price (e.g., \$/kgCO$_2$), converting the carbon footprint into monetary terms. Similarly, $w_{\text{th}}$ is the thermal penalty weight defined in \$/$^{\circ}$C, representing the economic risk associated with exceeding the safe operating temperature $T_{\text{core,max}}$.

\subsubsection{System Constraints}
To formulate a clear control problem, we organize the system constraints into three categories: (i) deterministic physics and operational constraints (state evolution, actuator bounds, energy balances); (ii) a robust objective term minimizing the worst-case expected cost; and (iii) a distributionally robust hotspot risk constraint (DR-CVaR on the peak temperature proxy, detailed in Section III-B).
\begin{itemize}

\item{\em Thermal Entity Dynamics:} The core temperature evolution must follow the RC model, driven by the uncertain aggregated IT load $P_{\text{IT},z}[t]$:

\vspace{-3mm}
{\small
\begin{equation}
    T_{\text{core},z}[t+1] = A_z T_{\text{core},z}[t] + B_z P_{\text{IT},z}[t] + C_z T_{\text{in},z}[t], \; \forall z \in \mathcal{L}_z
    \label{eq:dy}
\end{equation}
}

\item{\em Cooling System Energy Balance:} The heat removed by the CRAC units must equal the thermal load serviced by the chiller plant. The air-side heat removal is:

\vspace{-2mm}
{\small
\begin{equation}
    Q_{\text{air-side}}[t] = \sum_{i \in \mathcal{I}} m_{\text{air},i}[t] c_{p,\text{air}} (T_{\text{ret},i}[t] - T_{\text{sup},i}[t]).
\end{equation}
}
The thermal load met by the chillers is:

\vspace{-2mm}
{\small
\begin{equation}
    Q_{\text{refrigeration-side}}[t] = P_{\text{chiller}}[t] \cdot \text{COP}[t].
\end{equation}
}
This forms the key coupling constraint:

\vspace{-2mm}
{\small
\begin{equation}
    Q_{\text{air-side}}[t] = Q_{\text{refrigeration-side}}[t].
    \label{eq:energy_balance}
\end{equation}
}

\item{\em Operational and Safety Limits:} We enforce physical and operational limits within the model. The hard temperature limit in (\ref{eq:safety_limit}) acts as a control-oriented proxy constraint.

\vspace{-2mm}
{\small
\begin{equation}
    T_{\text{core},z}[t] \le T_{\text{core,crit}}, \quad \forall z \in \mathcal{L}_z.
    \label{eq:safety_limit}
\end{equation}
}

To bridge server-level risks with the zone-level optimization, we first identify the critical servers using real-time telemetry:

\vspace{-2mm}
{\small
\begin{equation}
    \mathcal{L}_{\text{hot}}[t] = \text{Top-}M_{\text{hot}} \{ T_{\text{tele},s}[t] \}_{s \in \mathcal{L}_s},
\end{equation}
}
where $M_hot$ is the number of critical servers and $T_{\text{tele},s}[t]$ denotes the real-time temperature measurement acquired directly from the telemetry channel of an individual server $s$ at time step $t$.

Next, we map these servers to their hosting zones:

\vspace{-2mm}
{\small
\begin{equation}
    Z_{\text{hot}}[t] = \{g(s) : s \in \mathcal{L}_{\text{hot}}[t]\}.
\end{equation}
}
Finally, we define a hotspot proxy variable for the optimization:

\vspace{-2mm}
{\small
\begin{equation}
    T_{\text{hot}}^{\text{proxy}}[t] = \max_{z \in Z_{\text{hot}}[t]} (T_{\text{core},z}[t] + \Delta_z),
\end{equation}
}
where $\Delta_z$ is a calibrated zone-to-server hotspot margin from historical statistics (e.g., the 95th percentile of the temperature difference $T_{\text{tele},s} - T_{\text{core},z}$ per zone).

\item{\em Decision Variable Bounds:} All control actions must remain within their physical operating limits.

\vspace{-2mm}
{\small
\begin{equation}
    0 \le s_i[t] \le 1, \quad \forall i \in \mathcal{I},
\end{equation}
\begin{equation}
    T_{\text{chw,set,min}} \le T_{\text{chw,set}}[t] \le T_{\text{chw,set,max}} .
    \label{eq:bounds}
\end{equation}
 }    
\end{itemize}
Finally, we summarize the CDRO model. The complete problem is defined as:

\vspace{-2mm}
{\small
\begin{equation*}
\begin{aligned}
    \text{(CDRO)}: \quad \min_{\mathbf{u}} \quad & \text{Eq. (\ref{eq:obj})} \\
    \text{s.t.} \quad & \text{Physical Dynamics: (\ref{eq:server_rc}) -- (\ref{eq:tower})} \\
    & \text{Operational Constraints: (\ref{eq:dy}) -- (\ref{eq:bounds})}
\end{aligned}
\end{equation*}
}
This formulation integrates the non-linear physical realities of AI data centers directly into the control task. It provides a formulation that combines deterministic plant constraints with distributionally robust hotspot risk control.

\vspace{-2mm}

\section{Methodology}
\label{sec:methodology}

In this section, we propose a two-stage framework integrating statistical learning with non-convex optimization to address contextual uncertainty and real-time constraints. The offline phase constructs high-fidelity uncertainty models and calibrates safety parameters. This setup enables the online phase to execute a fast and decomposed optimization algorithm for robust control.
\vspace{-2mm}
\subsection{Context-Aware Ambiguity Set Construction}
\label{ssec:ambiguity_set}

The core of our framework is the ambiguity set $\mathcal{P}(\mathbf{c}[t])$. It captures the uncertainty of the parameters based on the real-time context $\mathbf{c}[t]$. AI workloads and grid signals have complex patterns that change over time. To handle this, we use a residual-based strategy. We use a forecaster $f_{\text{forecast}}$ (using XGBoost) to find the deterministic parts of exogenous signals (workload, weather, and/or grid signals). We model the realization $\bm{\xi}[t]$ as a point forecast plus a residual: $\bm{\xi}[t] = \hat{\bm{\xi}}[t] + \bm{e}[t]$. The ambiguity set is defined over the distribution of these residuals using the Wasserstein metric.

\subsubsection{Context Features and k-NN Retrieval}
The context vector $\mathbf{c}[t]$ is important for capturing the distribution of residuals. We define $\mathbf{c}[t]$ to include real-time values and indicators of volatility. All features come from the online information set $\mathcal{I}_t$. For example, we use the variance of job duration estimates. For grid features, we use the spread between the Day-Ahead price and the last observed Real-Time price:

\vspace{-2mm}
{\small
\begin{equation}
    \Delta_{\text{spread}}[t] = C_{\text{DA}}[t] - C_{\text{RT}}[t-1],
\end{equation}
}
where $C_{\text{DA}}[t]$ and $C_{\text{RT}}[t]$ denote the DA and RT electricity prices, respectively. We use a k-nearest neighbor (k-NN) approach in the feature space. We retrieve a set of historical residuals $\{e_j\}_{j=1}^k$ from similar contexts to build an empirical distribution. We only use residuals from the training and validation sets to prevent using future data.

\subsubsection{Radius Calibration via Safety-Driven Backtesting}
A major challenge is choosing the Wasserstein radius $\rho(\mathbf{c}[t])$. Standard theory often fails because AI residuals have heavy tails. We use a data-driven approach instead. We define the ambiguity set as:

\vspace{-2mm}
{\small
\begin{equation}
    \mathcal{P}_{\rho}(\mathbf{c}[t]) := \left\{ \mathbb{P} \in \mathcal{M}(\Xi) \mid W_1(\mathbb{P}, \hat{\mathbb{P}}_{k}(\mathbf{c}[t])) \le \rho(\mathbf{c}[t]) \right\}.
\end{equation}
}
We group the context space into a set of $R$ volatility regimes $\{C_r\}_{r=1}^R$. The radius function is piecewise constant:

\vspace{-2mm}
{\small
\begin{equation}
    \rho(\mathbf{c}) = \rho_r^*, \quad \text{if } \mathbf{c} \in C_r, \; r \in \{1,\dots,R\}.
\end{equation}
}
For each regime, the optimal radius $\rho_r^*$ is the minimum value that satisfies the safety requirement on a validation set $\mathcal{D}_{\text{val}}$:

\vspace{-2mm}
{\small
\begin{equation}
    \rho_r^* := \min \left\{ \rho \ge 0 \mid \text{CVaR}_{\varepsilon}(V_t(\rho)) \le \text{Target}, \forall t \in \mathcal{D}_{\text{val}} \cap C_r \right\},
\end{equation}
}
where $V_t(\rho)$ denotes the empirical thermal violation magnitude under radius $\rho$, and $\text{Target}$ is the predefined safety threshold. Offline, we train the forecaster and k-NN index, compute residuals, partition the context into volatility regimes, and calibrate the minimum $\rho_r^*$ satisfying the target on validation data. The resulting context-calibrated radius, denoted $\rho(\mathbf{c}[t])$, is applied uniformly across both the robust objective and DR-CVaR safety constraints.

\begin{remark}[Value of Contextual Information]
By grouping residuals into regimes, we group similar data points. The uncertainty in each regime is lower than the total uncertainty. This allows the controller to use a smaller radius. A smaller radius reduces cooling costs while keeping the same safety level.
\end{remark}

\subsection{Tractable Reformulation via Nested Decomposition}
\label{ssec:reformulation}

The primal problem formulated in Section III involves an inf–sup optimization over an infinite-dimensional space of probability measures, coupled with non-convex physical constraints. This is computationally difficult to solve for real-time control. We employ a nested solution framework. We use an inner dualization to handle the stochastic uncertainty and an outer decomposition to handle the physical complexity.

\subsubsection{Step 1: Strong Duality Reformulation}
We first convert the stochastic worst-case expectation into a deterministic convex problem. The primal problem seeks to minimize the expected cost under the worst-case distribution within the Wasserstein ball. To address this infinite-dimensional challenge, we invoke the strong duality theory for Wasserstein DRO established by Esfahani and Kuhn \cite{EsfahaniKuhn2018}. By adapting their general theoretical result to our specific objective function, which is convex in decision variables and affine in uncertain parameters (i.e., $\ell(\mathbf{x},\bm{\xi}) = \mathbf{a}(\mathbf{x})^\top \bm{\xi} + d(\mathbf{x})$), we derive the following tractable counterpart.

\begin{proposition}[Tractable Counterpart for Objective]
\label{prop:dual_reformulation}
For a fixed decision vector $\mathbf{x}$ (representing cooling actions), the worst-case expected cost is equal to the optimal value of the following deterministic minimization problem:

\vspace{-2mm}
{\small
\begin{equation}
\label{eq:dual_form}
\begin{aligned}
\min_{\lambda, s_j} \quad & \lambda \rho(\mathbf{c}[t]) + \mathbf{a}(\mathbf{x})^\top \hat{\bm{\xi}} + d(\mathbf{x}) + \frac{1}{k} \sum_{j=1}^k s_j \\
\text{s.t.} \quad & s_j \ge \mathbf{a}(\mathbf{x})^\top \mathbf{e}_j, \quad \forall j \in \{1,\dots,k\}, \\
& \|\mathbf{a}(\mathbf{x})\|_* \leq \lambda.
\end{aligned}
\end{equation}
}
Since our cost function is affine in the uncertainty $\bm{\xi}$, the dual norm constraint simplifies to $\|\mathbf{a}(\mathbf{x})\|_* \le \lambda$.
\end{proposition}

\begin{proof}
This proof adapts the standard duality derivation in \cite{EsfahaniKuhn2018}. We start with the primal problem:

\vspace{-2mm}
{\small
\begin{equation}
    J_{wc}(\mathbf{x}) = \sup_{P \in \mathcal{P}_\rho(\mathbf{c}[t])} \int_{\Xi} \left( \mathbf{a}(\mathbf{x})^\top (\hat{\bm{\xi}} + \mathbf{e}) + d(\mathbf{x}) \right) P(d\mathbf{e}).
\end{equation}
}
We introduce a Lagrange multiplier $\lambda \ge 0$ for the Wasserstein constraint $W_1(P, \hat{P}_N) \le \rho(\mathbf{c}[t])$. The Lagrangian is defined as:

\vspace{-2mm}
{\small
\begin{equation}
    L(P, \lambda) = \mathbb{E}_P[\ell(\mathbf{x}, \bm{\xi})] + \lambda (\rho(\mathbf{c}[t]) - W_1(P, \hat{P}_N)).
\end{equation}
}
Since the ambiguity set contains the empirical distribution (Slater's condition holds), strong duality applies, allowing us to swap the supremum and infimum:

\vspace{-2mm}
{\small
\begin{equation}
    J_{wc}(\mathbf{x}) = \inf_{\lambda \ge 0} \left\{ \lambda \rho(\mathbf{c}[t]) + \sup_{P} \left( \mathbb{E}_P[\ell(\mathbf{x}, \bm{\xi})] - \lambda W_1(P, \hat{P}_N) \right) \right\}.
\end{equation}
}
By the Kantorovich-Rubinstein duality theorem, the inner supremum over distributions $P$ simplifies to an expectation over the empirical distribution $\hat{P}_N = \frac{1}{k} \sum_{j=1}^k \delta_{\mathbf{e}_j}$. This transforms the expectation into a finite average of specific suprema:

\vspace{-2mm}
{\small
\begin{equation}
    \frac{1}{k} \sum_{j=1}^k \sup_{\mathbf{e} \in \Xi} \left\{ \mathbf{a}(\mathbf{x})^\top (\hat{\bm{\xi}} + \mathbf{e}) + d(\mathbf{x}) - \lambda \|\mathbf{e} - \mathbf{e}_j\|_1 \right\}.
\end{equation}
}
For each historical sample $\mathbf{e}_j$, let $\mathbf{u} = \mathbf{e} - \mathbf{e}_j$. The inner maximization becomes $\mathbf{a}(\mathbf{x})^\top \mathbf{e}_j + \sup_{\mathbf{u}} \{ \mathbf{a}(\mathbf{x})^\top \mathbf{u} - \lambda \|\mathbf{u}\|_1 \}$. By the definition of the dual norm, $\mathbf{a}(\mathbf{x})^\top \mathbf{u} \le \|\mathbf{a}(\mathbf{x})\|_\infty \|\mathbf{u}\|_1$. To prevent the supremum from diverging to $+\infty$, we must enforce the dual constraint $\|\mathbf{a}(\mathbf{x})\|_\infty \le \lambda$. Under this condition, the supremum is bounded at $0$. Introducing auxiliary epigraph variables $s_j \ge \mathbf{a}(\mathbf{x})^\top \mathbf{e}_j$ to bound the worst-case realization for each sample directly yields the finite-dimensional linear program in \eqref{eq:dual_form}.
\end{proof}

Proposition \ref{prop:dual_reformulation} yields a tractable robust objective. 
To enforce thermal safety, we next reformulate the risk constraint under dual data granularity: the dynamics are modeled at the zone level, while hotspot risk is observed at the server level. 
We therefore map server telemetry to a zone-level proxy state.

\begin{definition}[Hotspot risk channel under dual granularity]
Given server telemetry at time $t$, let the hotspot server set be $\mathcal{L}_{\text{hot}}[t]\subseteq\mathcal{L}_s$ and define the mapped hotspot-zone set
$Z_{\text{hot}}[t]=\{g(s): s\in\mathcal{L}_{\text{hot}}[t]\}$.
We impose the DR-CVaR constraint on the proxy temperature $T_{\text{hot}}^{\text{proxy}}[t]$ constructed from $Z_{\text{hot}}[t]$.
\end{definition}

We apply the DR-CVaR constraint to $T_{\text{hot}}^{\text{proxy}}[t]$ to control worst-case hotspot risk.
Since $T_{\text{hot}}^{\text{proxy}}[t]$ involves a max operator over mapped zones, we adopt a local convex approximation of its response with respect to the uncertainty residuals, which enables a standard dual reformulation and yields a tractable constraint.

\begin{proposition}[Tractable DR-CVaR Safety Constraint]
\label{prop:dr_cvar_tractable}
The distributionally robust safety constraint $\sup_{P \in \mathcal{P}_\rho(\mathbf{c}[t])} \text{CVaR}_{\varepsilon}(T_{\text{hot}}^{\text{proxy}} - T_{\text{core,crit}}) \le 0$ admits the following tractable conservative counterpart using auxiliary variables $\eta \in \mathbb{R}$, $\lambda_T \ge 0$, and slack variables $u_j \ge 0$ for each sample $j \in \{1,\dots,k\}$:

\vspace{-2mm}
{\small
\begin{subequations}
\label{eq:dr_cvar_tractable}
\begin{align}
    \eta + \frac{1}{\varepsilon} \left( \lambda_T \rho(\mathbf{c}[t]) + \frac{1}{k} \sum_{j=1}^k u_j \right) &\le 0, \\
    u_j \ge (T_{\text{hot}}^{\text{proxy}}(\hat{\xi} + e_j) - T_{\text{core,crit}}) - \eta, \quad &\forall j \in \{1,\dots,k\} \label{eq:cvar_hinge1},\\
    u_j \ge 0, \quad \forall j \in \{1,\dots,k\}, \\
    \|\nabla_e T_{\text{hot}}^{\text{proxy}}\|_*(\mathbf{c}[t]) \le \lambda_T \label{eq:cvar_dual_norm}.
\end{align}
\end{subequations}
}
where $T_{\text{hot}}^{\text{proxy}}(\hat{\xi} + e_j)$ is the peak hotspot proxy temperature evaluated under the specific historical residual scenario $e_j$.
\end{proposition}

\begin{proof}
The proof applies the same strong duality principles used in Proposition \ref{prop:dual_reformulation} to the CVaR function. CVaR can be expressed as a minimization over $\eta$ of expectations involving the hinge loss function $h(\mathbf{x}, e) = (T_{\text{hot}}^{\text{proxy}}(\mathbf{x}, e) - T_{\text{core,crit}} - \eta)^+$. Under the assumption that the proxy temperature response is locally convex in the residual $e$, the hinge loss remains convex. Thus, the worst-case expectation $\sup_{P} \mathbb{E}_P [h(\mathbf{x}, e)]$ admits a dual representation involving the empirical average of the hinge loss plus a regularization term $\lambda_T \rho(\mathbf{c}[t])$. The constraints in (\ref{eq:cvar_hinge1})--(\ref{eq:cvar_dual_norm}) enforce this dual upper bound to be non-positive. This yields a conservative but tractable counterpart.
\end{proof}

\subsubsection{Step 2: Non-Convex ADMM Decomposition}
Even after the deterministic reformulation, the problem remains a non-convex MISOCP due to the underlying physics of the cooling plant, specifically the cubic fan power laws ($P \propto s^3$) and the bi-quadratic chiller COP curves. To solve this efficiently, we decompose the global problem into two smaller distinct subproblems. We use a Data Hall Subproblem ($\mathcal{X}$) and a Central Plant Subproblem ($\mathcal{Z}$). These are coupled only by the total thermal load variable $Q$.

We form the Augmented Lagrangian function:

\vspace{-2mm}
{\small
\begin{align}
    \mathcal{L}_\rho(\mathbf{x}, \mathbf{z}, y) = &F_{DRO}(\mathbf{x}) + G_{DRO}(\mathbf{z}) \\&+ y(Q_x - Q_z) + \frac{\rho_{admm}}{2} \|Q_x - Q_z\|^2 \nonumber,
\end{align}
}
where $F_{DRO}$ and $G_{DRO}$ are the robust objectives derived in \eqref{eq:dual_form} for the respective subsystems. The Alternating Direction Method of Multipliers (ADMM) iteratively solves these subproblems.

For the Data Hall Optimization in the $\mathbf{x}$-update step, we aim to determine the optimal CRAC fan speeds $s_i[t]$ that minimize fan power and thermal risk. Crucially, this subproblem now incorporates the tractable DR-CVaR safety constraint defined in Proposition \ref{prop:dr_cvar_tractable} to enforce distributionally robust hotspot risk control in the data hall subproblem. The cubic power term $P_{\text{fan},i} = c_f \cdot s_i^3$ is approximated using a Piecewise Linear (PWL) function by partitioning the domain $[0, 1]$ into $M_{\text{pwl}}$ segments with breakpoints $\{v_0, \dots, v_{M_{\text{pwl}}}\}$. By introducing binary variables $\delta_{m,i}[t]$ and continuous variables $\alpha_{m,i}[t]$, we enforce the relationship:

\vspace{-2mm}
{\small
\begin{equation}
\begin{cases}
s_i[t] = \sum_{m=1}^{M_{\text{pwl}}} \left( v_{m-1}\delta_{m,i}[t] + \alpha_{m,i}[t] \right), \\[1mm]
P_{\text{fan},i}[t] = c_f \sum_{m=1}^{M_{\text{pwl}}} \left( v_{m-1}^3 \delta_{m,i}[t] 
+ \frac{v_m^3 - v_{m-1}^3}{v_m - v_{m-1}} \alpha_{m,i}[t] \right), \\[1mm]
0 \le \alpha_{m,i}[t] \le \delta_{m,i}[t](v_m - v_{m-1}), \quad \forall m=1,\dots,M_{\text{pwl}}, \\[1mm]
\sum_{m=1}^{M_{\text{pwl}}} \delta_{m,i}[t] = 1, \quad \delta_{m,i}[t] \in \{0,1\}, \quad \forall m=1,\dots,M_{\text{pwl}}.
\end{cases}
\end{equation}
}
This formulation enables us to solve the subproblem:

\vspace{-2mm}
{\small
\begin{equation}
    \mathbf{x}^{k+1} \leftarrow \arg\min_{\mathbf{x} \in \mathcal{X}} \left( F_{DRO}(\mathbf{x}) + \mathbf{y}^k Q_x + \frac{\rho_{admm}}{2} \|Q_x - Q_z^k\|^2 \right)
\end{equation}
}
using standard commercial solvers (using Gurobi), as it is transformed into a convex MISOCP subject to the reformulated robust safety constraints.

Subsequently, for the Central Plant Optimization in the $\mathbf{z}$-update step, the objective is to optimize the chiller setpoint $T_{chw}$. The non-convexity here stems from the bi-linear relation in chiller power: $P_{ch} \cdot \text{COP} = Q_{load}$. Let $w = P_{ch} \cdot \text{COP}$. To maintain tractability, we relax this non-convex constraint by replacing the bilinear term $w$ with its convex envelope, known as the McCormick relaxation. This standard technique bounds the variable $w$ within a convex polyhedral set defined by the variable bounds, effectively replacing the non-convex equality with a set of linear inequality constraints. This relaxation transforms the Central Plant subproblem into a tractable convex program, which is iteratively tightened within the ADMM loop:

\vspace{-2mm}
{\small
\begin{equation}
    \mathbf{z}^{k+1} \leftarrow \arg\min_{\mathbf{z} \in \mathcal{Z}} \left( G_{DRO}(\mathbf{z}) - \mathbf{y}^k Q_z + \frac{\rho_{admm}}{2} \|Q_x^{k+1} - Q_z\|^2 \right).
\end{equation}
}

Finally, the dual variable $y$ is updated to enforce the energy balance constraint $Q_x = Q_z$.

\vspace{-2mm}
{\small
\begin{equation}
    y^{k+1} \leftarrow y^k + \rho_{admm}(Q_x^{k+1} - Q_z^{k+1}).
\end{equation}
}

The complete online execution process, integrating the reformulated robust objectives and the ADMM decomposition scheme, is formally presented in Algorithm \ref{alg:admm_cdro_loop}.

\setlength{\textfloatsep}{4pt plus 1pt minus 2pt}
\begin{algorithm}[t]
\renewcommand{\baselinestretch}{1.1}\small
\captionsetup{font=small}
\caption{\small{Online Solution via Nested Duality and ADMM}}
\label{alg:admm_cdro_loop}
\begin{algorithmic}[1]
\REQUIRE Current context $\mathbf{c}[t]$, calibrated radius $\rho(\mathbf{c}[t])$
\ENSURE Optimal control actions $\mathbf{u}^*[t] = (\mathbf{x}^*, \mathbf{z}^*)$
\STATE Update server telemetry, identify $\mathcal{L}_{\text{hot}}[t]$, and map to $Z_{\text{hot}}[t]$ via $g(\cdot)$.
\STATE Initialize $\mathbf{x}^0, \mathbf{z}^0, \mathbf{y}^0, Q_{\text{couple}}^0$.
\STATE \textbf{for} iteration $k = 0, 1, 2, \dots$ until convergence \textbf{do}
\STATE \quad \textit{// Solve Data Hall Subproblem}
\STATE \quad Formulate $\mathbf{x}$-subproblem using \eqref{eq:dual_form} and Prop. \ref{prop:dr_cvar_tractable}.
\STATE \quad Apply PWL approximations to non-convex constraints.
\STATE \quad $\mathbf{x}^{k+1} \gets \text{SolveMISOCP}(\mathbf{z}^k, \rho(\mathbf{c}[t]), \dots)$
\STATE \quad \textit{// Solve Central Plant Subproblem}
\STATE \quad Formulate $\mathbf{z}$-subproblem using \eqref{eq:dual_form}.
\STATE \quad Apply PWL approximations to chiller/tower curves.
\STATE \quad $\mathbf{z}^{k+1} \gets \text{SolveMISOCP}(\mathbf{x}^{k+1}, \dots)$
\STATE \quad \textit{// Update Consensus and Dual Variables}
\STATE \quad Update coupling variable $Q_{\text{couple}}^{k+1}$ and dual $\mathbf{y}^{k+1}$.
\STATE \quad Check primal/dual residuals for convergence.
\STATE \textbf{end for}
\STATE \textbf{return} $\mathbf{u}^*[t] = (\mathbf{x}^{k+1}, \mathbf{z}^{k+1})$

\end{algorithmic}
\end{algorithm}
\vspace{-4mm}
\subsection{Convergence and Complexity Analysis}
\label{ssec:convergence}

\subsubsection{Convergence of Non-Convex ADMM}
Standard ADMM convergence proofs typically rely on the convexity of the objective functions, which does not hold for the cubic fan power laws and bi-quadratic chiller curves inherent to our model. However, our problem structure aligns with the class of non-convex problems analyzed by Wang et al. \cite{WangYinZeng2019}. Our problem is a multi-block optimization coupled by linear constraints. By verifying that our objective functions are coercive and that the coupling constraints satisfy Lipschitz continuity, we can adapt their theoretical framework to establish convergence to a stationary point. This property is crucial for real-time deployment because it prevents unbounded optimization oscillations and supports stable closed-loop implementation.

\subsubsection{Computational Complexity and Scalability}
The online execution solves two MISOCP subproblems per ADMM iteration. For a realistic deployment, the optimization size depends on the prediction horizon $H$, CRAC count $|I|$, thermal entities $|\mathcal{L}_z|$, k-NN sample size $k$, and the number of PWL segments $M_{\text{pwl}}$. The binary variables strictly arise from the PWL approximations, scaling predominantly as $N_{\text{bin}} \approx H \cdot |I| \cdot M_{\text{pwl}}$. Continuous variables scale roughly as $N_{\text{cont}} \approx c_1 H|I| + c_2 H|\mathcal{L}_z| + c_3 Hk$. This structural property, where the variable count grows linearly with the prediction horizon and spatial dimensions, ensures theoretical tractability for real-time control intervals.

\vspace{-2mm}


\section{Case Studies and Numerical Results}
\label{sec:experimental_evaluation}

To thoroughly evaluate the performance of the proposed CDRO framework, we developed a detailed simulation model of an AI data center and subjected it to a series of stress tests. These scenarios are designed to assess the controller's ability to manage the trilemma of thermal safety, economic cost, and carbon footprint under the deep uncertainty typical of grid-interactive AI infrastructure. Specifically, we validate how the proposed DR-CVaR safety constraint effectively limits thermal risks while enabling flexibility.
\vspace{-4mm}
\subsection{Simulation Environment and Data}
\label{ssec:simulation_env}

The experiments use an EnergyPlus Python co-simulation with a 5-minute control interval ($\Delta t = 5$ min). We model a high-density data hall equipped with $|I|=4$ CRAC units. The IT heat load $P_{\text{IT}}[t]$ uses Google Cluster Data traces with added random bursts. Queue features $f_{\text{queue}}$ are constructed directly from this trace data. These specific features include the rolling variance of job durations, a proxy for queue length, and a burstiness indicator based on load increments. Exogenous inputs include Houston TMY3 weather, 5-minute ERCOT LMPs \cite{ERCOT_NP6_788_LMP_5min}, and WattTime carbon-intensity data \cite{WattTime_API_V3, WattTime_AOER_Signal}.

Data is split chronologically into Training (60\%, for XGBoost forecasting and k-NN), Validation (20\%, for offline radius calibration), and Testing (20\%, for reported results). During test replay, the controller strictly accesses the available filtration $\mathcal{I}_t$ using rolling forecasts $\hat{\bm{\xi}}$.

The nominal online controller is configured with a prediction horizon $H=12$, $k=30$ nearest neighbors, and $M_{\text{pwl}}=5$ piecewise-linear segments. Optimization is solved using Gurobi 10.0 on an i7-13700F processor over $|\mathcal{L}_z| = 10$ controlled thermal entities for scalability. Concurrently, $|\mathcal{L}_s| = 200$ server telemetry channels are monitored strictly for hotspot identification and risk evaluation via a fixed mapping $g:\mathcal{L}_s \to \mathcal{L}_z$. We define the hotspot set $\mathcal{L}_{\text{hot}}[t]$ as the top $M_{\text{hot}}=10$ servers, applying the DR-CVaR constraint on the mapped proxy with a tail probability $\varepsilon=0.05$.

\begin{table}[t]
\centering
\caption{Rolling-window performance statistics over 30 non-overlapping 72-hour test windows and 5 random seeds (mean $\pm$ 95\% CI). TCO and emissions are percentage differences relative to deterministic MPC. Risk metrics are evaluated on server telemetry.}
\label{tab:overall_comparison}
\renewcommand{\arraystretch}{1.2}
\setlength{\tabcolsep}{4pt}
\begin{tabular}{@{}lcccc@{}}
\toprule
\textbf{Controller} 
  & \textbf{EVP(\%)} 
  & \textbf{TVI($^\circ$C$\cdot$h)}
  & $\boldsymbol{\Delta}$\textbf{TCO(\%)}
  & $\boldsymbol{\Delta}$\textbf{Emis.(\%)} \\
\midrule
PID           & $8.6 \pm 1.4$  & $36.8 \pm 6.9$ & $12.7 \pm 1.9$ & $10.5 \pm 1.6$ \\
MPC-Det.      & $3.3 \pm 0.5$  & $15.8 \pm 2.4$ & \textbf{Baseline} & \textbf{Baseline} \\
Static RO     & $0.03 \pm 0.02$ & $0.4 \pm 0.3$ & $18.1 \pm 1.4$ & $14.8 \pm 1.2$ \\
SP            & $1.3 \pm 0.4$  & $5.8 \pm 1.4$  & $4.4 \pm 0.8$  & $3.9 \pm 0.7$  \\
\midrule
Min-Max MPC   & $0.05 \pm 0.03$ & $0.5 \pm 0.4$ & $16.0 \pm 1.2$ & $13.2 \pm 1.0$ \\
CC-MPC        & $0.92 \pm 0.27$ & $4.4 \pm 0.9$ & $3.2 \pm 0.6$  & $2.4 \pm 0.5$  \\
CVaR-MPC      & $0.61 \pm 0.20$ & $3.1 \pm 0.8$ & $3.7 \pm 0.7$  & $2.9 \pm 0.6$  \\
NC-DRO        & $0.10 \pm 0.05$ & $0.8 \pm 0.5$ & $6.2 \pm 0.9$  & $4.9 \pm 0.8$  \\
\midrule
\textbf{CDRO} & $\mathbf{0.07 \pm 0.04}$ & $\mathbf{0.6 \pm 0.4}$ & $\mathbf{2.3 \pm 0.5}$ & $\mathbf{1.9 \pm 0.4}$ \\
\bottomrule
\end{tabular}
\end{table}
\vspace{-4mm}

\subsection{Experimental Design}
\label{ssec:design}

We compare CDRO against deterministic baselines (PID, Deterministic MPC) and uncertainty-aware controls configured as follows:
\begin{itemize}
    \item \textit{Hard-constrained Robustness}: Static RO uses historical extreme box sets; Min-Max MPC tightens constraints using the 99th percentile of training residuals.
    \item \textit{Nominal Risk-awareness}: SP optimizes over empirical samples; CC-MPC and CVaR-MPC enforce empirical risk targets ($\varepsilon=0.05$) based on historical error distributions.
    \item \textit{Distributional Robustness}: Standard DRO uses a globally fixed Wasserstein radius $\bar{\rho}^\star$ calibrated on the validation set without contextual differentiation.
\end{itemize}

For fairness, all optimization baselines use the same zone-level dynamic model ($\mathcal{L}_z$) and server-level risk evaluation ($\mathcal{L}_s$). We evaluate performance using a two-tier strategy:
\begin{itemize}
    \item \textit{Rolling-Window Statistics (Primary)}: $N=30$ non-overlapping 72-hour test windows, evaluated across $S=5$ random seeds (for AI burst generation) to yield the statistical metrics and confidence intervals in Table \ref{tab:overall_comparison}.
    \item \textit{Stress Scenarios (Qualitative)}: Three specific 72-hour events (AI Workload Spike, Extreme Climate, Grid Volatility) to physically interpret transient controller behaviors.
\end{itemize}
\vspace{-4mm}

\subsection{Performance Metrics}
\label{ssec:metrics}

To formally evaluate the rolling-window outcomes, we define the window-level metrics for any given window $W_i$. 

Thermal risk is quantified by the window Empirical Violation Probability (EVP) and the Thermal Violation Integral (TVI). Based on the dual granularity design, these risk metrics are evaluated strictly on the server-level telemetry to reflect the true thermal state. Let $T_{\text{hot}}^{\text{tele}}[t] = \max_{s \in \mathcal{L}_s} T_{\text{tele},s}[t]$ denote the peak server temperature observed at time $t$. Operationally, $T_{\text{core,max}}$ is the recommended operating limit (SLA threshold) used for continuous TVI penalty accumulation, while $T_{\text{core,crit}}$ is the strict hardware protection limit used to define emergency EVP violations:
  
\vspace{-4mm}
{\small
\begin{align}
    \text{EVP}(W_i) &= \frac{1}{|W_i|} \sum_{t \in W_i} \mathbb{I}(T_{\text{hot}}^{\text{tele}}[t] > T_{\text{core,crit}}) ,\\
    \text{TVI}(W_i) &= \sum_{t \in W_i} \max(0, T_{\text{hot}}^{\text{tele}}[t] - T_{\text{core,max}}) \Delta t.
\end{align}
}

For economic and environmental impacts, the total operational cost is computed as $\text{TCO}(W_i) = \sum_{t \in W_i} C_{\text{elec}}[t] \cdot E_{\text{total}}[t]$. 
Crucially, while the optimization objective in (\ref{eq:objective_full}) uses a monetized carbon cost ($\kappa_{\text{CO}_2} \cdot CI \cdot E$) to drive decisions, our reported environmental metric strictly tracks the \textit{physical} carbon emissions $\text{Emissions}(W_i) = \sum_{t \in W_i} CI_{\text{grid}}[t] \cdot E_{\text{total}}[t]$ in kgCO$_2$. This decoupling ensures that the reported environmental benefits are not artifacts of arbitrary carbon pricing parameters.

\vspace{-3mm}
\subsection{Performance Analysis}
\label{ssec:results_analysis}

Table \ref{tab:overall_comparison} summarizes the statistical performance across the rolling-window evaluation. Confidence Intervals (CIs) are computed using a paired bootstrap method across the $N \times S$ window-seed samples to rigorously support our claims.

\subsubsection{Attribution of Performance Gains}
Table \ref{tab:overall_comparison} highlights two primary sources of CDRO's performance gains. First, distributional robustness is essential for safety. Nominal methods like CVaR-MPC exhibits a substantially higher out-of-sample EVP than CDRO under distribution shift. In contrast, CDRO improves robustness to distributional shifts in the evaluated scenarios by optimizing against a Wasserstein ambiguity set, achieving a mean EVP of 0.07\%. Second, context awareness drives economic efficiency. While NC-DRO also achieves near-zero empirical violations (EVP) in our test suite, its static radius incurs a $+6.2\%$ cost premium. CDRO dynamically shrinks its radius during low-volatility contexts, yielding a highly significant cost reduction compared to NC-DRO (Wilcoxon signed-rank $p < 0.01$). Finally, CDRO reduces the robustness premium by 13.7 percentage points compared with Min–Max MPC by using probabilistic DR-CVaR bounds instead of rigid worst-case constraint tightening, allowing it to safely ride thermal margins.

\subsubsection{Robustness to Workload Uncertainty (Scenario 1)}
To qualitatively explain these statistical results, Scenario 1 exposes the transient fragility of forecast-dependent controls during a sudden workload burst. As illustrated in Fig.~\ref{fig:heatwave_results}, deterministic MPC fails to anticipate the burst magnitude. CVaR-MPC mitigates the violation but remains sensitive to the heavy-tail residuals, resulting in minor safety breaches. NC-DRO and CDRO achieve near-zero empirical hotspot violations in this scenario by using distributionally robust objective optimization and DR-CVaR hotspot risk control. However, due to its fixed global conservatism, NC-DRO incurs unnecessary pre-cooling costs by ramping up long before the burst. In contrast, CDRO leverages the job queue context to detect rising volatility and expands its ambiguity set "just-in-time", effectively buffering the thermal shock with minimal wasted energy.

\begin{figure}[t]
    \centering
\includesvg[width=0.5\textwidth]{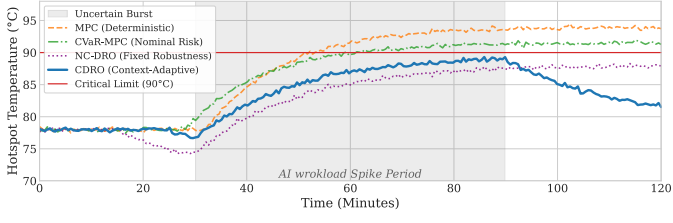}
\caption{Thermal response during an AI workload spike. CDRO proactively maintains safety without the excessive conservatism of standard DRO.}
    \label{fig:heatwave_results}
        \vspace{-4mm}
\end{figure}

\subsubsection{Resilience Under Physical Constraints (Scenario 2)}
Fig.~\ref{fig:scenario2} presents the trade-off between cost (TCO) and risk (EVP) for all methods during a heatwave scenario. CDRO consistently occupies the ideal region. Static methods like RO and Min-Max MPC fall into the ``High Cost'' region. Nominal risk methods (CVaR/CC-MPC) offer intermediate performance but do not consistently achieve the low empirical EVP levels attained by CDRO. Compared directly with these advanced baselines, CDRO achieves lower operational costs at a matched level of risk due to its context-adaptive radius and precise residual retrieval.

\begin{figure}[t]
    \centering
 \includesvg[width=0.47\textwidth]{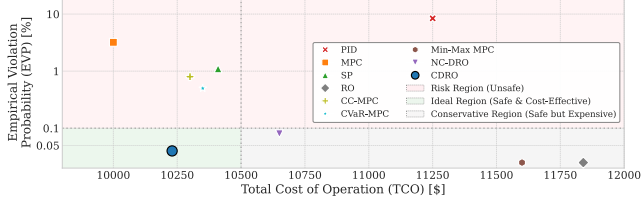}
    \caption{Pareto trade-off for Heatwave (TCO vs. EVP). CDRO outperforms uncertainty-aware baselines (points) by achieving the lowest cost at a safe risk level.}
    \label{fig:scenario2}
    \vspace{-2mm}
\end{figure}

\subsubsection{Grid-Interactive Decision Making (Scenario 3)}
Figure \ref{fig:scenario3} illustrates power response to a price spike. Unlike Min-Max MPC, which maintains a rigid safety buffer, CDRO's context-aware ambiguity set allows it to identify periods of internal stability. During these windows, it tightens the robustness radius $\rho$, allowing the system to reduce cooling power (shed load) to exploit high electricity prices without violating conditional risk limits. This explains its significant cost advantage over NC-DRO, which relies on a fixed robustness radius and cannot relax constraints to fully exploit price arbitrage.

\begin{figure}[t]
    \centering
 \includesvg[width=0.5\textwidth]{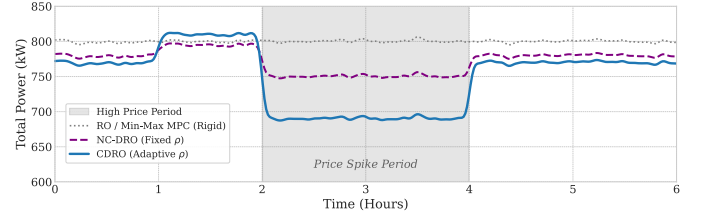}
    \caption{Power consumption during grid volatility. CDRO safely maximizes load shedding during high-price periods compared to static DRO.}
    \label{fig:scenario3}
        \vspace{-5mm}
\end{figure}
\vspace{-4mm}
\subsection{Aggregate Discussion and Sensitivity Analysis}
\label{ssec:discussion}

\subsubsection{Radius Calibration and Sensitivity}
We validated the effectiveness of our data-driven calibration approach. Figure \ref{fig:calibration_curve} illustrates the trade-off between conservatism (Cost) and safety (Risk) as the Wasserstein radius $\rho$ varies.
The solid lines represent the Context-Aware approach, while the dashed lines represent the Non-Contextual approach.
CDRO achieves the target safety level at a significantly lower average radius than NC-DRO. This visually confirms the statistical findings from the rolling-window evaluation: contextual information functionally translates into "cheaper" robustness.

\begin{figure}[t]
    \centering
\includesvg[width=0.5\textwidth]{calibration_comparison}
\caption{Regime-wise radius calibration. Contextual differentiation enables CDRO to guarantee safety at a lower average robustness cost.}
    \label{fig:calibration_curve}
        \vspace{-2mm}
\end{figure}

\subsubsection{Sensitivity to Risk Preference (\texorpdfstring{$\varepsilon$}{epsilon})}
A key contribution of our work is the ability to explicitly tune the safety-cost trade-off. We performed a sensitivity analysis by varying the DR-CVaR risk level $\varepsilon \in \{0.01, \dots, 0.25\}$. As shown in Fig.~\ref{fig:risk_sensitivity}, the Risk-Cost Pareto Frontier exhibits a convex shape. CDRO operates at the optimal "knee point," efficiently navigating the non-linear relationship between acceptable risk and operational expenditure.

\begin{figure}[t]
    \centering
 \includesvg[width=0.48\textwidth]{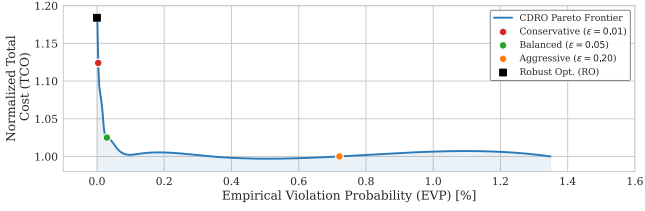}
\caption{Risk-Cost Pareto frontier. Varying the DR-CVaR tail parameter $\varepsilon$ provides explicit tuning between empirical safety and operational expenditure.}
    
    \label{fig:risk_sensitivity}
        \vspace{-5mm}
\end{figure}

\subsubsection{Robustness under Forecast Degradation}
Figure \ref{fig:noise_sensitivity} compares performance under varying forecast noise levels. Deterministic MPC and Nominal Risk methods (e.g., CVaR-MPC) degrade rapidly as forecasts worsen. In contrast, both CDRO and NC-DRO retain low empirical violation rates under forecast degradation relative to nominal baselines by using their calibrated ambiguity sets. However, CDRO maintains this out-of-distribution robustness at a strictly lower cost than NC-DRO across the entire noise spectrum.

\begin{figure}[t]
    \centering
\includesvg[width=0.48\textwidth]{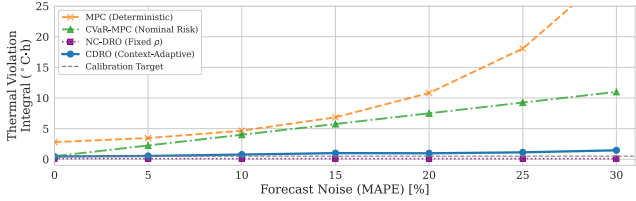}
    \caption{Empirical violation probability under forecast degradation. CDRO sustains calibrated safety targets against increasing noise, outperforming nominal methods.}
    \label{fig:noise_sensitivity}
\vspace{-1mm}
\end{figure}
\subsubsection{Real-time Feasibility and Scalability}
To ensure the CDRO controller is viable for standard 5-minute ($\Delta t$) dispatch intervals, we evaluated its computational scalability. The nominal configuration requires a median solve time of only 12.4 seconds. A small scaling sweep confirms that the runtime grows near-linearly with the horizon and spatial dimensions within practical ranges: when the prediction horizon $H$ doubles from 12 to 24, the median solve time increases from 12.4s to 25.1s; increasing the CRAC count $|I|$ from 4 to 8 yields 26.5s; and increasing the binary-intensive PWL segments $M_{\text{pwl}}$ from 5 to 7 yields 22.8s (all with fixed $k=30$). This scaling behavior strictly validates the controller's real-time feasibility for deployment in high-density data centers.

\vspace{-2mm}
\section{Conclusion}
\label{sec:conclusion}
This paper proposed a CDRO framework to address thermal management challenges in grid-interactive AI data centers. By applying a distributionally robust CVaR constraint on the critical hotspot channel, the controller dynamically adjusts its robustness radius based on real-time computational and grid volatility signals. This context-aware approach overcomes the limitations of brittle deterministic forecasts and overly conservative static bounds. High-fidelity rolling-window simulations demonstrate that CDRO effectively mitigates forecast error risks, empirically achieving near-zero thermal violations on critical hotspot servers in our test replay. Simultaneously, by relaxing conservatism during periods of low uncertainty, it reduces the operational cost premium of robustness by approximately 13.7 percentage points compared to standard min-max approaches.

\vspace{-5mm}

\bibliographystyle{IEEEtran}
\bibliography{bibtex/IEEEexample}

\end{document}